\documentclass[11pt]{amsart}
\usepackage[margin=1.05in]{geometry}
\usepackage{enumerate,bbm}
\usepackage[dvips]{epsfig}
\usepackage{graphicx}
\numberwithin{equation}{section}
\theoremstyle{plain}
\newtheorem{theorem}{Theorem}
\newtheorem{lemma}[theorem]{Lemma}
\newtheorem{corol}[theorem]{Corollary}

\newtheorem{prop}[theorem]{Proposition}
\theoremstyle{definition}

\newtheorem{example}[theorem]{Example}
\theoremstyle{remark}

\newcommand{\suppress}[1]{}

\newcommand{\Rplus}{\mathbb R_+}

\numberwithin{equation}{section}

\newcommand{\M}{{\text{sd}}}
\newcommand{\be}{\begin{equation}}
\newcommand{\ee}{\end{equation}}
\newcommand{\bea}{\begin{eqnarray}}
\newcommand{\eea}{\end{eqnarray}}
\newcommand{\bean}{\begin{eqnarray*}}
\newcommand{\eean}{\end{eqnarray*}}

\DeclareMathOperator{\support}{support}

\usepackage{color}
\newcommand{\ljs}[1]{{ \textcolor{red}{(LJS:  #1)}}{}}

\title{Allocation of Divisible Goods
under Lexicographic Preferences}

\author{Leonard J. Schulman}
\address{Caltech, MC305-16, Pasadena CA 91125, USA.}
\email{schulman@caltech.edu}
\author{Vijay V. Vazirani}
\address{College of Computing, Georgia Institute of
  Technology, Atlanta GA 30332, USA.}
\email{vazirani@cc.gatech.edu}
\date{\today .
 \textit{Author affiliations:} Caltech, \tt{schulman@caltech.edu}; Georgia Tech, \tt{vazirani@cc.gatech.edu}}
\begin{document}
\pagenumbering{gobble}
\clearpage
\thispagestyle{empty}

\begin{abstract}
We present a simple and natural non-pricing mechanism for
allocating divisible goods among strategic agents having lexicographic
preferences.  
Our mechanism has favorable properties of strategy-proofness
(incentive compatibility). In addition (and even when extended to the case of Leontief bundles) it enjoys Pareto efficiency,
envy-freeness, and time efficiency.
\end{abstract}

\maketitle
\clearpage
\pagenumbering{arabic}
\section{Introduction}

The study of principled ways of allocating divisible goods
among agents has long been a central topic in mathematical
economics. The method of choice that emerged from this
study, the Arrow-Debreu market model~\cite{AD}, provides a
powerful approach based on pricing and leads to the
fundamental welfare theorems. However, these market-based
methods have limitations when agents are assumed to be
strategic, e.g., these methods are not incentive
compatible. Issues of the latter kind have been studied
within the area of mechanism design for the last four
decades, and have played a large role in the last decade in
algorithmic game theory~\cite{nisan}.

In this paper our primary focus is on deriving a non-pricing mechanism for allocating
divisible goods, that satisfies incentive-compatibility, Pareto
optimality and envy-freeness. A natural approach to achieving Pareto optimality and envy-freeness is to start in a greedy fashion
by assigning agents their most favored goods, and gradually moving on to their less favored choices.  
It is easy to come up with several ways of making this approach precise---two are described in Section \ref{counter}---and achieve Pareto optimality and envy-freeness. However, it is not a priori clear that it is possible to also achieve
incentive compatibility, without which a mechanism is of doubtful merit in an environment of strategic agents. In the main contribution of our paper we show that a third version of this approach, 
the \textit{Synchronized
Greedy (SG) mechanism}, achieves all three properties.

The SG mechanism can be seen as generalizing a mechanism introduced by Cr\`{e}s
and Moulin~\cite{CM01}, called {\em Probabilistic Serial (PS)}, in the context of a job scheduling
problem, and studied further by Bogomolnaia and
Moulin~\cite{BM01} for the allocation of indivisible goods\footnote{These
mechanisms for allocation of indivisible goods are randomized. Our
focus on divisible goods is just as general, since an allocation of
divisible goods can be used without further modification as a
randomized allocation of indivisible goods in the same quantities.}.
The preference model assumed by \cite{BM01} was first order stochastic dominance, which we will shorten to {\em sd-preference}.
They showed that in this model, PS is efficient, envy-proof and weakly incentive compatible. Furthermore, they
showed that in this model, no mechanism satisfies all three properties, i.e., efficiency, envy-proofness and incentive compatibility.
In view of the second result, we need to relax the model in order to obtain a mechanism satisfying all three properties; we
do so by resorting to the {\em lexicographic preference relation} and assuming that the goods are divisible.

Lexicographic preferences date back to the work of Hausner \cite{Hausner} and are of interest  to economists for the following
reasons. They yield a total order on the set of all allocations (unlike sd-preferences, say, which only form a partial order)
and they can be seen as a strong-preferences limit of von Neumann-Morgenstern utilities. A preference relation that is complete, 
transitive and satisfies the continuity condition that preferences between allocations are preserved under limits
is known to be representable by a utility function \cite{Collel}. Of these, lexicographic preferences 
forgo continuity.
\suppress{
For the setting defined below, we show that SG has favorable
efficiency, incentive compatibility, and fairness
properties. Our setting assumes that each agent has a
\textit{lexicographic} preference relation over goods.
 We note that
this preference relation is \textit{rational} in the sense that
it is complete and transitive. It does not, on the other
hand, satisfy
the continuity condition that preferences between
allocations are preserved under limits; a rational preference
relation that also satisfies this continuity condition
is known to be representable by a utility function,
see~\cite{Collel}.
}
What favorable properties can be achieved in the area of goods allocation using only non-pricing mechanisms is a difficult question. The present paper can be regarded as carving out a certain special case, namely the limit in which agents have very strong preferences among the goods, and providing strong positive guarantees in this case. 
In this limit there is an additional motivation to use non-pricing mechanisms, because very strong preferences might cause a pricing mechanism to do little more than ensure that the wealthiest agents get what they want. By focusing on non-pricing mechanisms, we can study what game-theoretic properties an allocation mechanism can achieve, without depending on what resources the agents possess or care to invest in the game. 

There are many every-day examples where something like our model comes up---naturally, not in market economy transactions, but in other societal mechanisms for allocation. An important class is allocation of public resources, e.g.,
placement lotteries in public schools, see Kojima~\cite{Kojima09} for further examples and references. (Note also that this kind of example employs a standard reduction of the indivisible goods case to the divisible goods case by randomization.) 

\suppress{
Another class of examples is coordination within teams. To illustrate, suppose Alice, Bob and Carol are camp counselors. They have to divide up many tasks in their 12 daily duty hours: leading canoeing activities, leading rock-climbing, meeting parents of prospective campers, kitchen duty, cleaning latrines, etc. Alice loves to rock-climb, but fears the water and prefers even latrine duty to canoeing. Bob loves canoeing, but will do anything to stay out of the kitchen where he once saw a rat. Carol enjoys cooking even for large groups, but cannot stifle her impatience with visiting parents. These modestly-paid counselors are not going to arrange their shared summer schedule with a monetary exchange. But they might use a protocol like the one we analyze in this article. Due to its strategy-proofness it can be implemented simply and transparently.
}

The recent paper of Saban and Sethuraman~\cite{SS} builds on our work and solves several open problems stated in an earlier version
of this paper \cite{SV}; these results are described at the end of Section~\ref{lit}. 
The broader challenge of the utility-functions version of the allocation problem remains largely open. The simplicity of the SG mechanism is perhaps encouraging toward the existence of allocation mechanisms maintaining favorable (maybe weaker) game-theoretic properties in this setting. Finally, we note that independent of our work, Cho \cite{Cho} has also studied the use of lexicographic preferences 
in the context of probabilistically assigning indivisible objects to agents.

\subsection*{Parameters of the problem}
In the allocation problem there are $m$ distinct divisible goods which need to be allocated
among $n$ agents. Good $j$ ($1 \leq j \leq m$)
is available in the amount $q_j > 0$, and agent $i$ $(1
\leq i \leq n$) is
to receive a specified $r_i > 0$ combined quantity of all goods;
the parameters satisfy $\sum_j q_j \geq \sum_i r_i$, i.e., the total supply is at least as large 
as the total demand. If this inequality fails, our mechanism may still be run after rescaling expectations so that each agent $i$ is to receive the quantity $r'_i = r_i (\sum q_j)/(\sum r_\ell)$. 
So in the sequel we may assume $\sum_j q_j \geq \sum_i r_i$.


\subsection*{Preferences: the non-Leontief case}
The \textit{non-Leontief} case of our problem is this. 
An
\textit{allocation} of goods is a list of
numbers $a_{ij} \geq 0$, with $\sum_j a_{ij}=r_i$ and
$\sum_i a_{ij} \leq q_j$, indicating that agent $i$ receives
quantity $a_{ij}$ of good $j$. The vector
$a_{i*}=(a_{i1},\ldots,a_{im})$ is referred to as agent
$i$'s (share of the) allocation. Each agent $i$ has a
\textit{preference list}, which is a permutation $\pi_i$ of
the goods; $(a_{i\pi_{i}(1)},\ldots,a_{i\pi_{i}(m)})$ is
agent $i$'s \textit{sorted allocation}. Agent $i$'s
preference among allocations is induced by
\textit{lexicographic order}. That is to say, agent $i$
\textit{lexicographic-prefers} $a_{i*}$ to $b_{i*}$, written $a_{i*}
>_i b_{i*}$, if the leftmost nonzero coordinate of
$(a_{i\pi_{i}(1)},\ldots,a_{i\pi_{i}(m)}) -
(b_{i\pi_{i}(1)},\ldots,b_{i\pi_{i}(m)})$ is positive.
Furthermore, we will say that agent $i$
prefers $a_{i*}$ to $b_{i*}$ in the stochastic domination order~\cite{BM01}, or 
\textit{sd-prefers}  $a_{i*}$ to $b_{i*}$, 
written $a_{i*} >^\M_i b_{i*}$,
 if
\[ \mbox{for all} \ k = 1, \ldots, m: \ \ \sum_{\ell=1}^k
a_{i\pi_{i}(\ell)} \geq \sum_{\ell=1}^k b_{i\pi_{i}(\ell)} ,\] with at least
one of the inequalities being strict. 
The symbols $\geq_i$ and  $\not \geq_i$
will have the obvious interpretations.

Since an agent's preferences depend only on his own share of the allocation, we speak interchangeably of an agent's preference for an allocation or an allocation share. In particular, $a_{i*}
>_i b_{i*}$ may be written more simply as $a >_i b$, and
$a_{i*} >^\M_i b_{i*}$ may be written as $a >^\M_i b$.

\subsection*{Preferences: Leontief Bundles}
Some of our results hold in the more general setting of lexicographic preferences among Leontief bundles, and some fail in that setting; details below.
A Leontief bundle is specified by a non-negative vector $\lambda=(\lambda_1,\ldots,\lambda_m) \in \Rplus^m$ (where $\Rplus=$ non-negative reals).
The set of goods $j$ for which $\lambda_j$ is positive is called the {\em support} of this bundle.
(If the set is of size one, we refer to this as a singleton bundle; in Economics this is sometimes also called the linear case.)
 If $q \in \Rplus^m$ then the bundle $\lambda$ may be allocated from $q$ in any quantity $\alpha \in \Rplus$ such that $\alpha \lambda_j \leq q_j$ for all $j$. 
\suppress{
 If $x$ is a 
Let $a_1, \ldots, a_m$ denote the proportions in which $m$ divisible goods are desired and let $x$ be a bundle of goods. Then the utility
of this bundle is defined to be
\[ \min \left\{ {x_1 \over a_1}, \ldots, {x_m \over a_m} \right\} \]
The set of goods $j$ for which $a_j$ is positive is called the {\em support} of this function.
}
In an instance of our problem, a list of $M$ Leontief bundles $\lambda^1,\ldots,\lambda^M$ is specified, including among them the $m$ singleton bundles (hence always $M\geq m$).
\suppress{ (There is little importance to upper bounding $M$, as in any case only few bundles will wind up being used in our allocations. \ljs{Revise this comment.}) }
It is convenient, and in our context sacrifices no generality, to impose the convention that for every bundle $\lambda^k$, $\sum_1^m \lambda^k_j=1$. 

The case $m=M$, in which all bundles are singletons, is of course a special case of the Leontief framework, but to distinguish it from the general situation we call it the ``non-Leontief'' case.

The framework we are concerned with is that
each agent $i$ has a preference list specified by a permutation $\pi_i$ of the bundles. A Leontief allocation is an $n \times M$ matrix $\ell$ in which $\ell_{ik}$ represents the quantity of bundle $k$ allocated to agent $i$. A Leontief allocation $l$ imposes the goods allocation $A(l)$, an $n \times m$ matrix, by $A(l)_{ij}=\sum_{k=1}^M l_{ik} \lambda^k_j$. We further require that a Leontief allocation satisfy the conditions $\sum_j A(l)_{ij}=r_i$ (thanks to the convention above this is equivalent to $\sum_k l_{ik}=r_i$)
and $\sum_i A(l)_{ij}\leq q_j$. We speak of $A(l)_{i*}$ and $l_{i*}$ as agent $i$'s \textit{share} of, respectively, the goods and the Leontief bundles.
The vector
 $(l_{i \pi_i(1)},\ldots,l_{i\pi_i(M)})$ is agent $i$'s \textit{sorted Leontief share}. 
Agent $i$'s
preference among allocations is induced by
\textit{lexicographic order} on his share of the allocation. That is to say, agent $i$
\textit{lexicographic-prefers} $l$ to $l'$, written $l
>_i l'$, if the leftmost nonzero coordinate of
$(l_{i\pi_{i}(1)},\ldots,l_{i\pi_{i}(M)}) -
(l'_{i\pi_{i}(1)},\ldots,l'_{i\pi_{i}(M)})$ is positive. Thus, for any goods allocation $a$, there is a favored Leontief allocation, denoted $L^\pi(a)$, defined by providing each agent with the best Leontief share that can be assembled from his share of the goods---to be explicit, this is obtained by starting with $a_{i*}$ as the available goods vector, and then, for $k$ from $1$ to $M$, setting
$L^\pi(a)_{i\pi_i(k)}$ to be the largest $\alpha$ such that $($(available goods vector)$-\alpha \lambda^k) \in \Rplus^M$, then subtracting $\alpha \lambda^k$ from the available goods vector and iterating.

We say that agent $i$ \textit{sd-prefers} allocation $a$ to $b$, written 
 $a >^\M_i b$,
 if
\[ \mbox{for all} \ K = 1, \ldots, M: \ \
\sum_{k=1}^K L^\pi(a)_{i \pi_i(k)} \geq
\sum_{k=1}^K L^\pi(b)_{i \pi_i(k)} ,
\] with at least
one of the inequalities being strict.  

\subsection*{The two orders}
Observe that
``lexicographic-prefers'' is a complete preference relation
without indifference contours (since it is antisymmetric for
distinct allocation shares), and that
``sd-prefers'' is an incomplete preference
relation; moreover the lexicographic order is a refinement
of the sd order, i.e., sd-prefers
implies lexicographic-prefers. The phrase ``agent $i$
weakly X-prefers'' will be used to include the possibility
that agent $i$'s share is identical in the two allocations.

\subsection{Our results}

The SG mechanism is deterministic, treats all agents
symmetrically, and has the following properties.

\subsubsection*{Properties w.r.t.\ sd preference:} 
\begin{itemize}

\item If all $r_i$'s are equal, the allocation produced by
  the SG mechanism in response to truthful bids is envy-free
  in the following sense: each agent weakly
  sd-prefers his allocation to that of any other
  agent. This holds also in the Leontief case.
\end{itemize}

\subsubsection*{Properties w.r.t.\ lexicographic preference:} 

\mbox{} 

(Since most of our paper deals with the relation
``lexicographic-prefers'', we subsequently abbreviate it to
``prefers''.)
\begin{itemize}
\item The allocation produced by the SG mechanism in response
to truthful bids is Pareto efficient.
This holds also in the Leontief case.
\item \label{ic} Incentive compatibility for a single agent: In the non-Leontief case, the SG
  mechanism is strategy-proof if $\min_j q_j \geq  \max_i r_i$. 

We give counterexamples (a) in the absence of this inequality, (b) for the Leontief case.
\item \label{icg} Generalizing the previous item, we have: Incentive compatibility for a coalition:
The SG mechanism is group strategy-proof
against coalitions of $\ell$ agents if $\min_j q_j \geq
\max_{S: |S|=\ell} \sum_{i \in S} r_i$. 

\item The running time to implement the SG mechanism is
  $\tilde{O}(mn)$ in the non-Leontief case, and $\tilde{O}(n(m^2+M))$ in the Leontief case.

\item
\textit{Any} Pareto efficient allocation can be produced using a suitable 
``variable speeds'' extension of the SG mechanism. This holds also in the Leontief case. (However, the variable speeds extension does not possess the rest of the properties listed above.)
\end{itemize}

The incentive compatibility properties 
are the main results of this paper.

\subsection{Literature}
\label{lit}
There has been considerable work on the strategy-proof
allocation of divisible goods in Arrow-Debreu economies,
starting with the seminal work of Hurwicz \cite{Hurwicz},
e.g., see
\cite{Dasgupta,Kato,Sonnen,Serizawa1,Serizawa2,Zhou91}.
Most of these results are negative, among the recent ones
being Zhou's result showing that in a 2-agent, $n$-good pure
exchange economy, there can be no allocation mechanism that
is efficient, non-dictatorial (i.e., both agents must
receive non-zero allocations) and strategy-proof
\cite{Zhou91}.

The paper that is most closely related to our work is that
of Bogomolnaia and Moulin~\cite{BM01}. In their setting
there are $n$ agents and $n$ indivisible goods, each agent
having a total preference ordering over the goods; the
desired outcome is a matching of goods with agents. A
straightforward mechanism for allocating one good to each
agent is \textit{random priority (RP):} pick a uniformly random
permutation of the agents and ask each agent in turn to
select a good among those left. It is easy to see that this
mechanism is \textit{ex post efficient}, i.e., the
allocation it produces can be represented as a probability
distribution over Pareto efficient deterministic
allocations, and it is strategy-proof. However, it is not ex
ante efficient.
A random allocation is said to \textit{ex ante efficient} if
for any profile of von Neumann-Morgenstern utilities that
are consistent with the preferences of agents, the expected
utility vector is Pareto efficient. It is easy to see that
ex ante efficiency implies ex post efficiency.

Solving a conjecture of Gale~\cite{Gale}, Zhou~\cite{Zhou90}
showed that no strategy-proof mechanism that elicits von
Neumann-Morgenstern utilities and achieves Pareto efficiency
can find a ``fair'' solution even in the weak sense of equal
treatment of equals. He further showed that the solution
found by RP may not be efficient if agents are endowed with
utilities that are consistent with their preferences. Hence,
ex ante efficiency had to be sacrificed, if
strategy-proofness and fairness were desired.

In the face of these choices, the work of Bogomolnaia and
Moulin gave the notion of \textit{ordinal efficiency} that
is intermediate between ex post and ex ante efficiency; an
allocation $a$ is ordinally efficient if there is no other
allocation $b$ such that every agent sd-prefers
$b$ to $a$. They went on to show that the mechanism called
\textit{probabilistic serial (PS)}, introduced in Cr\`{e}s
and Moulin~\cite{CM01}, yields an ordinally efficient
allocation. Further they show that PS is envy-free and
weakly strategy-proof, defined appropriately for the partial
order ``sd-prefers''. Finally, Bogomolnaia and
Moulin define an extension of PS by introducing different
``eating rates'' and show that this set of mechanisms
characterizes the set of all ordinally efficient
allocations.

Katta and Sethuraman~\cite{KS06} generalize the setting of
Bogomolnaia and Moulin to the ``full domain'', i.e., agents
may be indifferent between pairs of goods.
Thus, each agent partitions the goods by equality and
defines a total order on the equivalence classes of her
partition (the agent is equally happy with any good received
from an equivalence class). For this setting, they give a
randomized mechanism that is a generalization (different
from ours) of PS and achieves the same game-theoretic
properties as PS. 

A mechanism that probabilistically allocates indivisible
goods can also be viewed as one that fractionally allocates
divisible goods. 
Under the latter interpretation, the SG
mechanism is equivalent to PS for the case that $m = n$ and
the quantity of each good and the requirement of each agent
is one unit.  An important difference is that Bogomolnaia
and Moulin analyze PS under an incomplete preference
relation (stochastic dominance) in which ``most'' allocation shares
are incomparable; whereas we analyze SG under a complete
preference relation (lexicographic) that is a refinement of
stochastic dominance. The statement that a mechanism's allocation
is Pareto efficient w.r.t.\ lexicographic preferences is
considerably stronger than the same statement w.r.t.\
stochastic dominance preferences, because each agent's share is
dominated by more alternative shares in the lexicographic
order, than it is in the sd order; so, fewer
allocations are Pareto efficient in the lexicographic than in
the sd order.  Our results should be viewed
therefore as demonstrating that the PS mechanism and its
natural generalization, SG, have far stronger game-theoretic
properties than even envisioned in~\cite{BM01}.

For somewhat related questions primarily regarding exchange economies, see Barber\`{a} and Jackson~\cite{BarberaJ95}, Nicolo~\cite{Nicolo04}, Ghodsi et al.~\cite{Ghodsi11}, and Li and
Xue~\cite{Li}. Finally, we remark only that the problem of allocating a
\textit{single} divisible good among multiple agents with
known privileges is considerably different; the principal
issue studied in that problem is how to make the division in
a manner that is fair w.r.t.\ the given privileges. This is
known as the bankruptcy problem and has a long history,
e.g., see~\cite{ONeill82,AumannM85}. Despite an interesting
resemblance between the PS mechanism and some of the
mechanisms used in the solutions of that
problem~\cite{Kaminski00}, the issues at stake in the
bankruptcy literature are distinct from those in our paper
and its predecessors.

Saban and Sethuraman \cite{SS} solve some of the open problems stated in an earlier version of this paper. They consider the special case that all $r_i=1$. First they show that our condition $\min_j q_j \geq  \max_i r_i$ is tight in the sense that for any $q_1<1$ there exists an $n$, a finite list  $q_2,\ldots,q_n$, and agent preferences such that no mechanism is efficient, envy-free and strategyproof. They also show that if $q_1<1$, and list $q_2,\ldots,q_n$ and the agent preferences are given, then SG achieves all three properties if and only if any mechanism achieves all three properties. 
Finally for the generalized setting of Katta and Sethuraman, where agents can be indifferent between objects,
they show that no mechanism can satisfy all three properties. 

Since the PS rule is not strategyproof, recent work has studied the situation where agents are strategic. A Nash equilibrium for the PS rule is a 
preference profile for which no agent has an incentive to report a different profile. 
\cite{Aziz1} show that a pure Nash equilibrium is guaranteed to exist; however determining whether a given preference profile is 
a Nash equilibrium is coNP-complete.

\section{The Synchronized Greedy Mechanism}
The mechanism is simple. Each agent $i$ submits a preference
list $\sigma_i$. The
submitted list may or may not, of course, agree with his true preference
list $\pi_i$.

(A simple case to consider is that of $M = m = n$ and
all $q_j = r_i = 1$. Because of the restriction that each preference list must include
all $m$ singleton bundles, each agent's preference list in this case is a permutation of the $m$ goods. Despite being quite special, this case, or the slightly more general case in which $M=m \leq n$ and all $r_i$ are equal, is already interesting to analyze and is well motivated by the examples, mentioned earlier, involving sharing of tasks or of scarce public resources.)

The mechanism simulates the following physical process.
Consider each good $j$ as a ``liquid'', and each agent as a receptacle of capacity $r_i$. 
The mechanism starts out at time $0$ by (for all $i$ in parallel) pouring bundle $\lambda^{\sigma_i(1)}$ into receptacle $i$ at rate $r_i$
 units of
liquid per unit time. Each good $j$ is therefore being drained at rate $\sum_i r_i \lambda^{\sigma_i(1)}_j$.
(Note that since $\sum_j {\lambda_{j}} = 1$, the total liquid being added to receptacle $i$ per unit time is $r_i$, as desired.)

This continues until one of the
goods, say $j$, is exhausted. For all agents who were currently being allocated bundles with $j$ in their support, their favorite Leontief bundle has now been exhausted. (We say that a Leontief bundle has been \textit{exhausted} at a given time if any of the goods in its support has been exhausted, and otherwise that the bundle is {\em available}.)
All such agents, $i$, are immediately allocated the next available bundle on their preference list, and the pouring of bundles continues. The algorithm continues in this way, allocating to an agent from the next available bundle whenever the current bundle has been exhausted. Since the singleton bundles are included in all preference lists, all agents continuously receive goods at rate $r_i$ until time $1$, at which time they simultaneously complete their full allocation.

Observe that the Leontief allocation $l$ constructed by SG satisfies $l=L^\pi(A(l))$ because the bundles are provided to each agent greedily based on the availability of goods.

This continuous process can easily be converted into a
discrete algorithm with the run time cited earlier: maintain a priority
queue of goods, keyed by termination times. Each time a good is exhausted, each agent is assigned its next unexhausted bundle, and an updated termination time for each good is computed using the coefficients of the active bundles.

Observe that if an agent prefers bundle $\lambda$ to bundle $\lambda'$, and $\support(\lambda) \subseteq \support(\lambda')$, then $\lambda'$ may be removed from the agent's preference list. It cannot be allocated to the agent by SG nor can it be part of any Pareto efficient allocation to the agent.

\suppress{
{\bf Remark:} Observe that if $v_j$ is preferred to $v_l$ by agent $i$ and if the support of $v_j$ is a subset of the support of
$v_l$ then at the time $i$ considers $v_l$, some good in its support will be already exhausted and this choice will not be available. Hence, $i$ 
will never get any allocation corresponding to $v_l$. Thus $v_l$ is a vacuous entry in $i$'s preference list.
}

\section{Properties of the Synchronized Greedy Mechanism}
\subsection{Pareto Efficiency}
Let $l^{\sigma}$ be the allocation created by
the SG mechanism in response to bids $\sigma$ declared by the
agents. As before $\pi$ denotes the truthful bids.

\begin{theorem}
\label{thm.Pareto.newproof}
The allocation produced by the SG mechanism in response to
truthful bids is Pareto efficient w.r.t.\ lexicographic preference. That is to say, for all $l \neq l^\pi$, $\exists i \; l <_i l^\pi$.
\end{theorem} 

\begin{proof} 
For agent $i$ and for $K \geq 1$ let $t_{iK}= \frac{1}{r_i} \sum_{k=1}^K l^\pi_{i \pi_i(k)}$. If agent $i$ receives a positive quantity of his $K$'th-most-favored bundle, then $t_{iK}$ is the time when that bundle is exhausted in SG. If the agent receives nothing from the bundle then the bundle is exhausted in SG no later than $t_{iK}$.

Suppose for contradiction the existence of $l$ s.t.\ $\forall i \; l \geq_i l^\pi$, and for some $i$, $ l >_i l^\pi$. Let $t$ be minimum s.t.\ $ \exists i,K$ s.t.\
$t=t_{iK}< \frac{1}{r_i} \sum_{k=1}^K l_{i \pi_i(k)}$. Note, if $t_{i'K'}<t$ then $t_{i'K'}=\frac{1}{r_{i'}}\sum_{k=1}^{K'} l_{i' \pi_i'(k)}$.

For every one of the bundles $b \in \{\pi_i(1),\ldots,\pi_i(K)\}$ there is a good $j(b)$ that appears positively in $b$ and which is exhausted by time $t$. Since $t_{iK}< \frac{1}{r_i} \sum_{k=1}^K l_{i \pi_i(k)}$ while $t_{iK'}= \frac{1}{r_i} \sum_{k=1}^{K'} l_{i \pi_i(k)}$ for all $K'<K$, some agent $i' \neq i$ receives strictly less of good $j(\pi_i(K))$ in $l$ than in $l^\pi$. Since $j(\pi_i(K))$ is exhausted in SG by time $t$, this means that there is some $K''$ such that $\frac{1}{r_{i'}} \sum_{k=1}^{K''} l_{i' \pi_{i'}(k)} < \frac{1}{r_{i'}} \sum_{k=1}^{K''} l^\pi_{i' \pi_{i'}(k)} \leq t$. This contradicts the minimality of $t$.
\end{proof}

\subsection{Strategy-Proofness}
A mechanism is said to be \textit{strategy-proof} if for every agent
and for every list of bids by the remaining agents, the agent cannot
obtain a strictly improved allocation by lying.

\begin{theorem}
\label{thm.incentive-compatible}
\mbox{}
In the non-Leontief case, the SG mechanism is strategy-proof if $\min q_j \geq \max
r_i$.
\end{theorem}

\begin{proof}
Without loss of generality focus on agent $1$. For the remainder of
this proof $\pi_2,\ldots,\pi_n$ are arbitrary bids by the agents
$2,\ldots,n$, but $\pi_1$ is agent $1$'s truthful bid.
We need to
show that for any bid $\sigma_1$ (and write
$\sigma=(\sigma_1,\pi_2,\ldots,\pi_n)$), $a^\sigma_{1*}
\leq_1 a^\pi_{1*}$. The theorem is trivial if
$a^\sigma=a^\pi$.

The theorem is also trivial if agent $1$, bidding
truthfully, receives only his top choice. So we may suppose
that agent $1$ does not receive the entire allocation of any
one good.

We may also suppose that if $a^\sigma_{1j}=0$ and
$a^\sigma_{1j'}>0$, then $\sigma_1^{-1}(j)>\sigma_1^{-1}(j')$. (Define $\sigma_1^{-1}(j)$ to be the $s$ such that
  $\sigma_1(s)=j$. Define $\pi_1^{-1}(j)$ analogously.) In
other words, all the requests in $\sigma_1$ that come up
empty may as well be deferred to the end.

 Let $G(j)=\{j': \pi_1^{-1}(j') \leq
  \pi_1^{-1}(j)$ and $a^\pi_{1j'}>0\}$. These are the goods that agent $1$ weakly prefers to good $j$ and receives a positive quantity of in the allocation $a^\pi$.

Say that agent $1$ \textit{sacrifices} good $j$ in $\sigma$
if:
\begin{enumerate}
\item \label{1sac1} $a^\pi_{1j}>0$,
\item \label{1sac2} $\sigma_1^{-1}(j)>\left|G(j)\right|$, and
\item \label{1sac3}
$\pi_1^{-1}(j)<\pi_1^{-1}(j')$ if $j'$ also satisfies
  (\ref{1sac1}),(\ref{1sac2}).
\end{enumerate}
That is to say, $j$ is the most-preferred good which agent
$1$ receives a positive quantity of in $\pi$, but requests
later in $\sigma$ than in $\pi$.

For a collection of bids $\rho$ let $T^\rho_j$ be the
time at which good $j$ is exhausted if the mechanism is run
with bids $\rho$.

Agent $1$ must sacrifice some good, call it $B$, since
otherwise the allocation will not change. See
Figure~\ref{fig-strat-pf-pf}. We will show that agent $1$ receives
strictly less of $B$ in $\sigma$ than in $\pi$, and that this is not
compensated for by getting more of more-preferred goods.

\begin{figure}
\includegraphics[height=70mm]{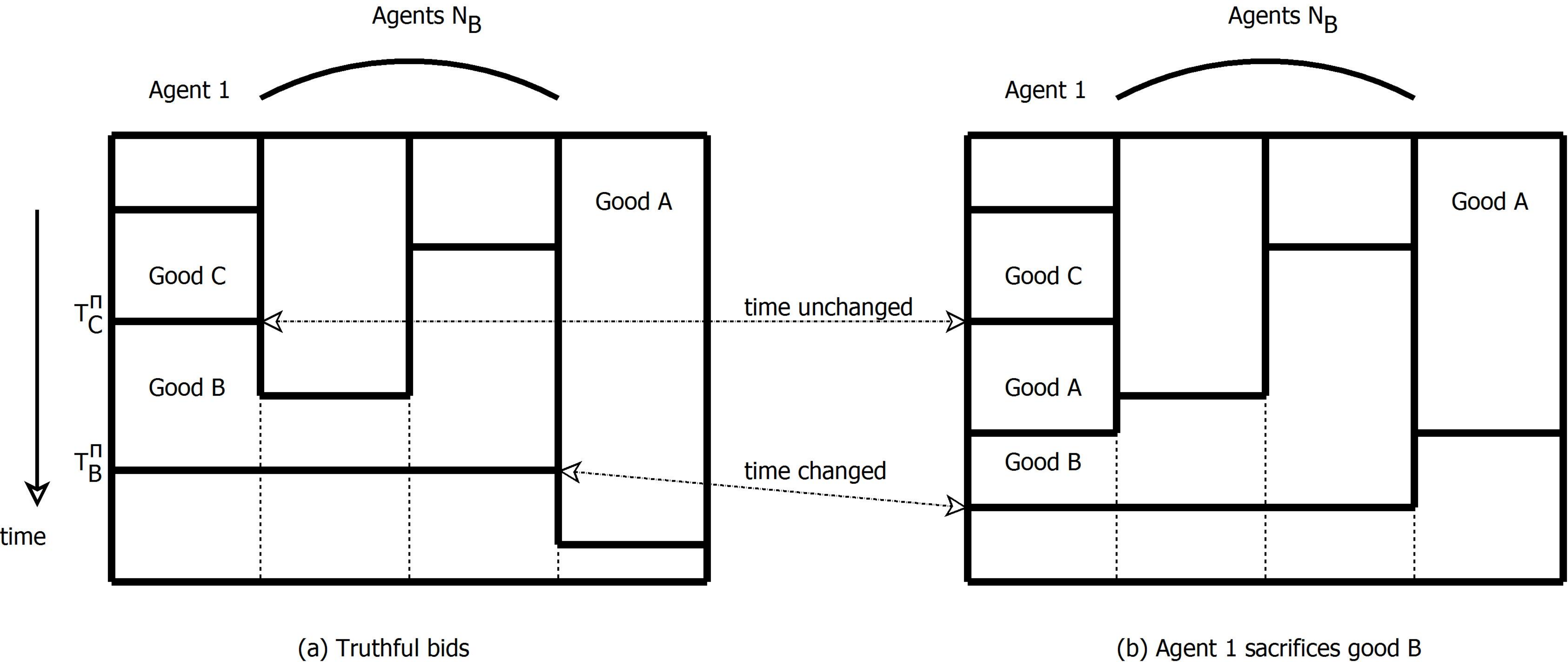}
\caption{The mechanism with truthful vs.\ lying bids of
  Agent 1}
\label{fig-strat-pf-pf}
\end{figure}

\begin{lemma} If $D$ is a good and
  $T^\pi_D<T^\pi_B$, then $T^\sigma_D \leq
  T^\pi_D$. \label{timing-lemma} \end{lemma}
\begin{proof}
Supposing the contrary, let $D$ be a counterexample
minimizing $T^\pi_D$. Since $T^\pi_D<T^\pi_B$,
 $D \neq B$.

Now let $i$ be any agent (who may or may not be agent $1$)
for whom $a^\pi_{iD}>0$. Due to the minimality of $D$, each
of the goods $j$ which $i$ prefers in $\pi$ to $D$, has
$T^\sigma_j \leq T^\pi_j$.  Therefore $i$ requests $D$ at a
time in $\sigma$ that is at least as soon as the time $i$
requests it in $\pi$.

Since this holds for all $i$ who received a positive
allocation of $D$ in $\pi$, the lemma follows.
\end{proof}
Let $N_B$ be the set of agents $i\neq 1$ for whom
$a^\pi_{iB}>0$. The condition on $r_i$'s and $q_j$'s ensures that this set is nonempty.

Due to the lemma, for each agent in $N_B$, the request time for $B$ in $\sigma$ is
weakly earlier than it is in $\pi$.
Now let $C$ be the good such that $\pi_1^{-1}(C)$ is maximal
subject to $\pi_1^{-1}(C)<\pi_1^{-1}(B)$ and $a^\pi_{1C}>0$. Due to the lemma, all goods $j'$ such that
$\pi_1^{-1}(j')\leq \pi_1^{-1}(C)$ have $T^\sigma_{j'} \leq
T^\pi_{j'}$. Next we show:
\begin{prop}
If $\pi_1^{-1}(j')\leq \pi_1^{-1}(C)$, then
$a^\sigma_{1j'}=a^\pi_{1j'}$. \end{prop}
\begin{proof} Supposing the contrary, let
$\pi_1^{-1}(j')$ be minimal such that $\pi_1^{-1}(j')\leq
  \pi_1^{-1}(C)$ and
$a^\sigma_{1j'}\neq
a^\pi_{1j'}$. There are two possibilities to consider.

(a) $a^\sigma_{1j'} < a^\pi_{1j'}$. This is not possible
 because then $a^\sigma_{1*} <_1 a^\pi_{1*}$.

(b) $a^\sigma_{1j'} > a^\pi_{1j'}$. Note:
\begin{lemma} Let $j_1,j_2$ be such that
$\pi_1^{-1}(j_1)\leq \pi_1^{-1}(B)$, $\pi_1^{-1}(j_2)\leq
\pi_1^{-1}(B)$, $a^\pi_{1j_1}>0$, and
$\pi_1^{-1}(j_1)<\pi_1^{-1}(j_2)$. Then
$\sigma_1^{-1}(j_1)<\sigma_1^{-1}(j_2)$. \label{insideLem1} \end{lemma}
\begin{proof} Consider the least $j_1$ that is part of a
  pair $j_1,j_2$ violating
  the lemma. Then $j_1$ satisfies conditions
  (\ref{1sac1}),(\ref{1sac2}) above, contradicting that $B$ is
  the good sacrificed by agent $1$. \end{proof}
It follows that $T^\sigma_{j'} \geq \sum_{j'':
  \pi_1^{-1}(j'') \leq  \pi_1^{-1}(j')} a^\sigma_{1j''}$.
Due to the minimality of $j'$, this means that if  $a^\sigma_{1j'} >
a^\pi_{1j'}$, then $T^\sigma_{j'} >T^\pi_{j'}$,
  contradicting our earlier conclusion. This completes
  demonstration of the Proposition. \end{proof}
A consequence of the Proposition is that $T^\sigma_C=T^\pi_C$.

Since agent $1$ sacrifices $B$, his request time for $B$ in
$\sigma$ is strictly greater than his request time
for $B$ in $\pi$.

Recall that $N_B$ is nonempty. At time $T^\pi_B$, the agents
of $N_B$ have received as least as much of $B$ in $\sigma$
as they have in $\pi$, and the latter is positive. On the
other hand, at the same time $T^\pi_B$, agent $1$ has
received strictly less of $B$ in $\sigma$ than he has in
$\pi$.
In order for agent $1$
to receive at least as much of $B$ in $\sigma$ as in
$\pi$, he would have to receive all of $B$ that is
allocated after time $T^\pi_B$; however, that is not
possible, because the set of agents receiving $B$ after $T^\pi_B$
includes $N_B$. Thus $a^\sigma_{1*} <_1 a^\pi_{1*}$.
\end{proof}

\subsection{Necessity of a Hypothesis on $\{r_i\},\{q_j\}$}
We next provide an example in which strategy-proofness
fails in the absence of the condition $\max r_i \leq \min
q_j$. For convenience now let $r_1 \geq \ldots \geq r_n$ and $q_1
\leq \ldots \leq q_m$.

\begin{example} \label{noStratPf}
Let $n=2$ and $m=3$. Let $r_1=r_2=3/2$; label the goods
$A,B,C$, let $q_A=q_B=q_C=1$, and let the
preference lists be $\pi_1=(A,B,C)$, $\pi_2=(B,C,A)$. If
agent $1$ bids truthfully he receives the sorted allocation
$(1, 0, 1/2)$.  If instead he bids $(B,A,C)$ (while agent
$2$ bids truthfully), he receives the improved sorted
allocation $(1, 1/2, 0)$. See Figure~\ref{fig-n2m3}.
\end{example}

\begin{figure}
\includegraphics[height=60mm]{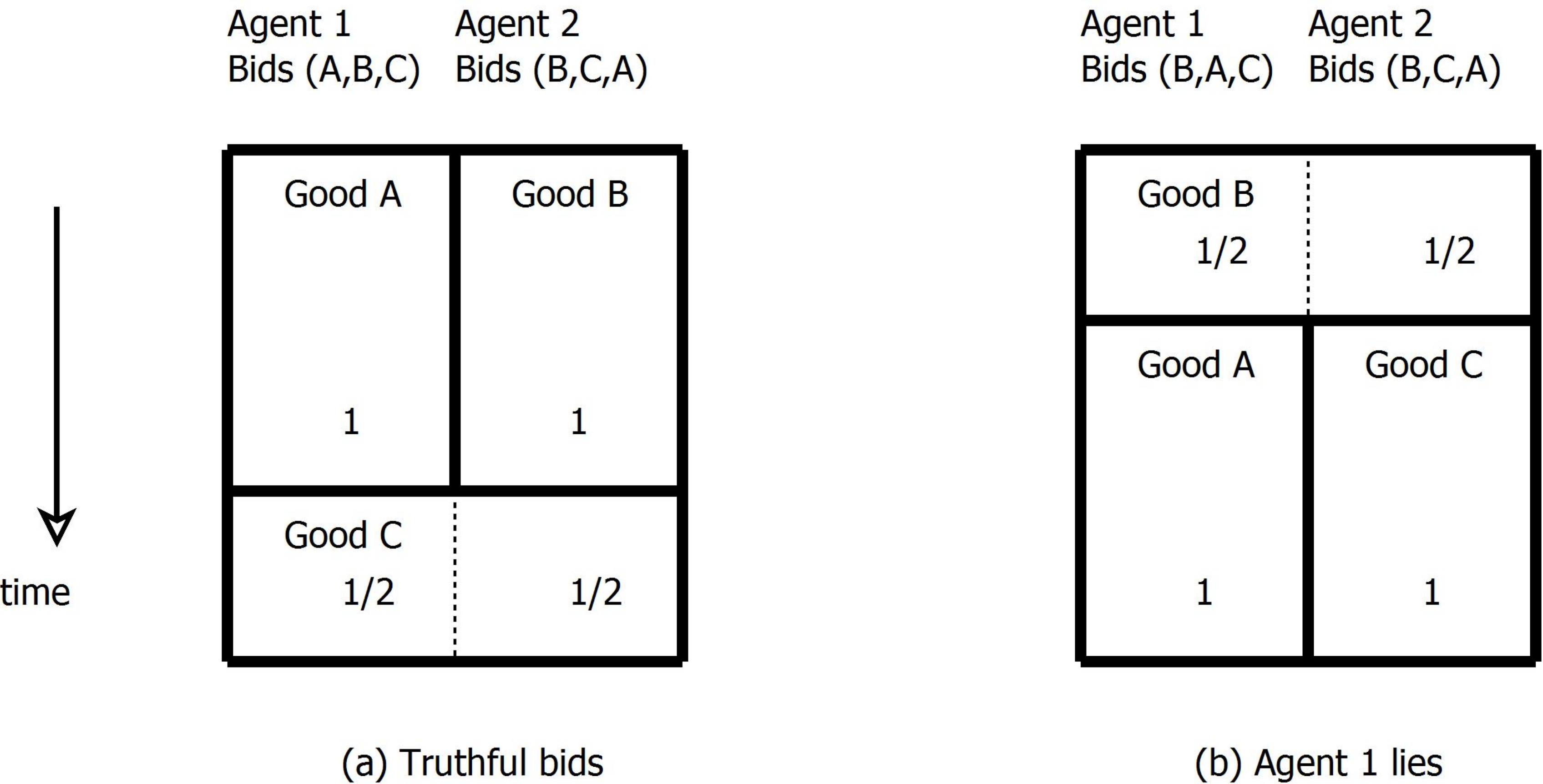}
\caption{Failure of strategy-proofness without the
  hypothesis of Theorem~\ref{thm.incentive-compatible}}
\label{fig-n2m3}
\end{figure}

This example does not limit the theorem sharply, because it
uses $r_1=(3/2) q_1$ rather than $r_1$ arbitrarily close to
$q_1$. Jeremy Hurwitz has pointed out that one may construct
similar examples whenever $r_1 \geq q_1/(1-q_2/\sum
q_j)$; this would appear to be a tight bound.

\subsection{Failure of strategy-proofness for the Leontief case}
Theorem~\ref{thm.incentive-compatible} has no equivalent for general
Leontief bundles. Consider the following four-agent system with 
$r_1=r_2=r_3=r_4=1$
and three goods in supply $q_A=q_B=1, q_C=2$.
Agent $1$'s desired
Leontief bundles are in the preference order $(A,B,C)$ (this agent is
interested only in singleton bundles); agent $2$ and $3$'s desired Leontief
bundles are in the order $(\frac{1}{2}A+\frac{1}{2}B, C, A, B)$; agent $4$'s Leontief bundles are in the order $(B,C,A)$.

Under
truthful bidding agent $1$ receives the sorted goods allocation $(1/2,
0, 1/2)$. By bidding instead $(B,A,C)$, agent $1$ receives the
improved sorted goods allocation $(2/3, 1/3, 0)$. See Figure~\ref{fig-n4m3}.

\begin{figure}
\includegraphics[height=50mm]{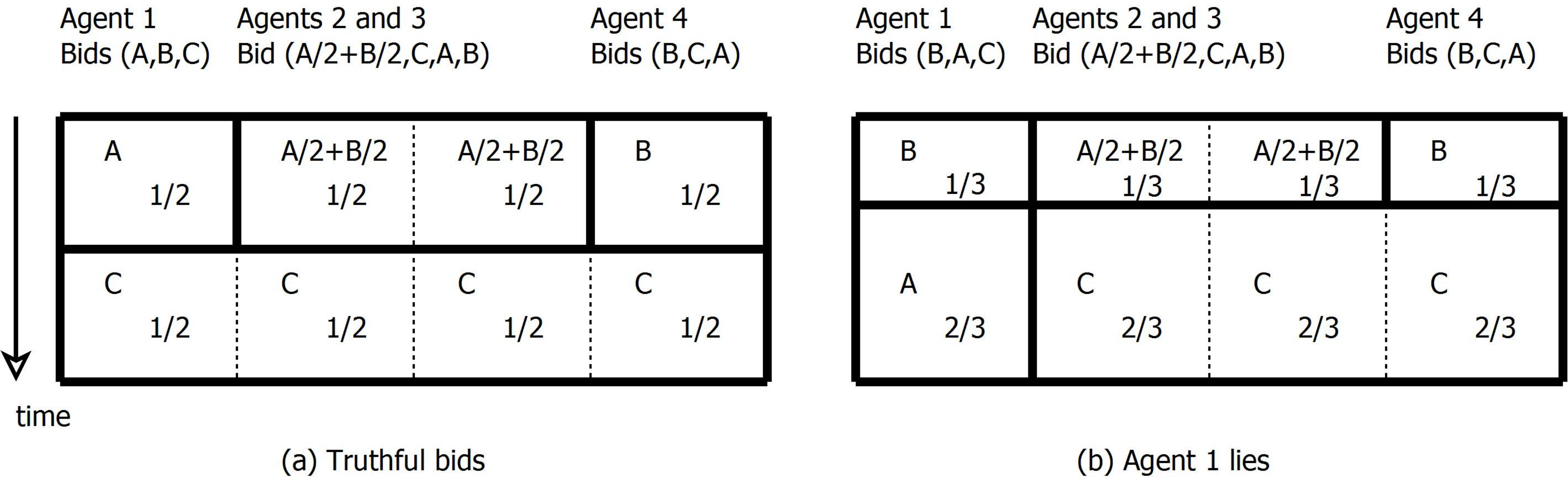}
\caption{Failure of strategy-proofness in the Leontief case}
\label{fig-n4m3}
\end{figure}

\subsection{Group Strategy-Proofness}
A mechanism is \textit{group strategy-proof} against a
family $F$ of subsets of agents if for every ``coalition''
$S \in F$ and for any list of bids by the agents outside of $S$,
the agents of $S$ cannot obtain an improved
allocation by lying, where by ``improved allocation'' we
mean that no agent of $S$ obtains a worse allocation and at
least one obtains a strictly better allocation.

We now provide the following generalization of Theorem~\ref{thm.incentive-compatible}:

\begin{theorem}
\label{thm.gruppen}
\mbox{}
In the non-Leontief case, the SG mechanism is group strategy-proof against the family of subsets
$S$ for which $\min_j q_j \geq \sum_{i \in S} r_i$.
\end{theorem}

\begin{corol}
\label{old.thm.gruppen}
\mbox{}
In the non-Leontief case, the SG mechanism is group strategy-proof against coalitions of
$\ell$ agents if $\min_j q_j  \geq \max_{S: |S|=\ell}
\sum_{i \in S} r_i$.
\end{corol}

The proof of Theorem~\ref{thm.gruppen} follows a structure similar to that of Theorem~\ref{thm.incentive-compatible} but the argument is complicated by the fact that 
 different agents in $S$
can sacrifice different goods, and some of the agents may actually be
better off due to their untruthful bids (as they may benefit
from the interactions among the several lies). The proof needs to
effectively ``chase through'' an unbounded iteration of good transfers
relative to $a^\pi$, and show that \textit{some}
agent in the coalition is worse off than in
$\pi$. Fortunately, this can be done without explicitly pursuing the iteration.

\begin{proof}
Let $S$ be a minimal counterexample. That is,
\begin{itemize}
\item[(a)]
$\min_j q_j \geq
\sum_{i \in S} r_i$;
\item[(b)] With $\pi_i$ representing in this proof the truthful preferences
for $i \in S$ and arbitrary preferences for $i \notin S$, there are
bids $\sigma_i$ for $i \in S$ such that every $i \in S$ ``is a willing
participant in the coalition $S$'', namely (with $\sigma_\ell=\pi_\ell$ for
$\ell \notin S$) $a^\sigma_{i*} \geq_i a^\pi_{i*}$;
\item[(c)] For some $i \in
S$, $a^\sigma_{i*} >_i a^\pi_{i*}$;
\item[(d)] No strict subset of $S$
satisfies (a),(b),(c).
\end{itemize}

Note by minimality that in $\sigma$, every agent $i \in S$ bids untruthfully
(differently from $\pi$) and this has an effect, namely, if $i$
reverts to bidding according to $\pi$ then the allocation is different
than in $\sigma$.

\suppress{
Let $S$ be a coalition as in the theorem statement. We need
to show there is no list of bids for the agents in $S$ such
that all do at least as well as in $\pi$, and some do
strictly better.
}

If $a^\pi_{i\pi_i(1)}=r_i$ for all $i \in S$, that is, with
truthful bids these agents receive only their top choices,
then none of them can be strictly rewarded by submitting a
different bid.

Otherwise (i.e., if $a^\pi_{i\pi_i(1)}<r_i$ for some $i \in S$),
then thanks to the hypothesis, under the truthful bids
$\pi$, every good has a positive allocation outside $S$.

We may simplify the argument slightly by supposing
that for each agent $i \in S$, if $a^\sigma_{ij}=0$ and
$a^\sigma_{ij'}>0$, then $\sigma_i^{-1}(j)>\sigma_i^{-1}(j')$. In
other words, all the requests that come up empty may as
well be deferred to the end.

Let $G(i,j)=\{j': \pi_i^{-1}(j') \leq
  \pi_i^{-1}(j)$ and $a^\pi_{ij'}>0\}$.

Say that agent $i$ \textit{sacrifices} good $j$ in $\sigma$ if:
\begin{enumerate}
\item \label{sac1} $a^\pi_{ij}>0$,
\item \label{sac2} $\sigma_i^{-1}(j)>\left|G(i,j)\right|$, and
\item \label{sac3}
$\pi_i^{-1}(j)<\pi_i^{-1}(j')$ if $j'$ also satisfies
  (\ref{sac1}),(\ref{sac2}).
\end{enumerate}
Some good must be sacrificed by some agent, since otherwise the
allocation will not change. (However, while every agent in
$S$ is untruthful, not every $i \in S$ necessarily sacrifices a
good; setting $\sigma_i(j)>\pi_i(j)$ might have an effect
even if $a^\pi_{ij}=0$ because of increased availability of
$j$ due to bidding changes of other agents.)

Of all the sacrificed goods let $B$ be one for which $T^\pi_B$
is minimal.

\begin{lemma} If $D$ is a good and
  $T^\pi_D<T^\pi_B$, then $T^\sigma_D \leq
  T^\pi_D$. \end{lemma}
\begin{proof}
Supposing the contrary, let $D$ be a counterexample
minimizing $T^\pi_D$.
By the minimality of $B$, $D$ cannot be a sacrificed good.

Now let $i$ be any agent (inside or outside of $S$) for whom
$a^\pi_{iD}>0$. Due to the minimality of $D$, each of the
goods $j$ which $i$ truthfully prefers to $D$, has
$T^\sigma_j \leq T^\pi_j$.  Therefore $i$ requests $D$ at a
time in $\sigma$ that is at least as soon as the time $i$
requests it in $\pi$.

Since this holds for all $i$ who
received a positive allocation of $D$ in $\pi$, the lemma
follows.
\end{proof}

Let $O_B \subseteq S$ be the set of agents who sacrifice
$B$, and let $N_B$ be the set of agents $i$ for whom
$a^\pi_{iB}>0$ but who do not sacrifice $B$. Due to the
lemma, for each agent in $N_B$, the request time
for $B$ in $\sigma$ is weakly earlier than it is in
$\pi$. Now consider an agent $i \in O_B$.  Let $C$ be the
good such that $\pi_i^{-1}(C)$ is maximal subject to
$\pi_i^{-1}(C)<\pi_i^{-1}(B)$ and $a^\pi_{iC}>0$. Due to the
lemma, all goods $j'$ such that $\pi_i^{-1}(j')\leq
\pi_i^{-1}(C)$ have $T^\sigma_{j'} \leq T^\pi_{j'}$. Next we
show:
\begin{prop}
If $\pi_i^{-1}(j')\leq \pi_i^{-1}(C)$, then
$a^\sigma_{ij'}=a^\pi_{ij'}$. \end{prop}
\begin{proof} Supposing the contrary, let
$\pi_i^{-1}(j')$ be minimal such that $\pi_i^{-1}(j')\leq
  \pi_i^{-1}(C)$ and
$a^\sigma_{ij'}\neq
a^\pi_{ij'}$. There are two possibilities to consider.

(a) $a^\sigma_{ij'} < a^\pi_{ij'}$. This is not possible
because $i$ is a willing participant in the coalition.

(b) $a^\sigma_{ij'} > a^\pi_{ij'}$. Note:
\begin{lemma} Let $j_1,j_2$ be such that
$\pi_i^{-1}(j_1)\leq \pi_i^{-1}(B)$, $\pi_i^{-1}(j_2)\leq
\pi_i^{-1}(B)$, $a^\pi_{ij_1}>0$, and
$\pi_i^{-1}(j_1)<\pi_i^{-1}(j_2)$. Then
$\sigma_i^{-1}(j_1)<\sigma_i^{-1}(j_2)$.  \label{insideLemi}
\end{lemma}
\begin{proof} Identical to the proof of
  Lemma~\ref{insideLem1} with agent $i$ in place of agent
  $1$.
\end{proof}
It follows that $T^\sigma_{j'} \geq \sum_{j'':
  \pi_i^{-1}(j'') \leq   \pi_i^{-1}(j')} a^\sigma_{ij''}$.
Due to the minimality of $j'$, this means that if  $a^\sigma_{ij'} >
a^\pi_{ij'}$, then $T^\sigma_{j'} >T^\pi_{j'}$,
  contradicting our earlier conclusion. This completes
  demonstration of the Proposition. \end{proof}
A consequence of the Proposition is that $T^\sigma_C=T^\pi_C$.

Since agent $i$ sacrifices $B$, his request time for $B$ in
$\sigma$ is strictly greater than his request time
for $B$ in $\pi$.

Since we are in the case that every good has a positive
allocation outside $S$, $N_B$ is nonempty. At time
$T^\pi_B$, the agents of $N_B$ have received as least as
much of $B$ in $\sigma$ as they have in $\pi$, and the
latter is positive. On the other hand, at the same time
$T^\pi_B$, the agents of $O_B$ have received strictly less
of $B$ in $\sigma$ than they have in $\pi$. In order for the
agents of $O_B$ to receive collectively at least as much of
$B$ in $\sigma$ as in $\pi$, they would have to receive all
of $B$ that is allocated after time $T^\pi_B$; however, that
is not possible, because the set of agents receiving $B$
after $T^\pi_B$ includes $N_B$. Therefore there is some $i
\in O_B$ for whom $a^\sigma_{iB} < a^\pi_{iB}$. This
contradicts the requirement that $i$ be a willing
participant in the coalition $S$.
\end{proof}

\begin{example} \label{noGpStratPf}
Example~\ref{noStratPf}, in which strategy-proofness failed
absent the hypothesis of
Theorem~\ref{thm.incentive-compatible}, can be extended in a
straightforward manner to one in which the group
strategy-proof property fails to hold absent the hypothesis
of Corollary~\ref{old.thm.gruppen}. Again use $m=3$, but
instead of two agents, use $n=2\ell$ agents, the first half
having the same preference order $(A,B,C)$ as agent $1$ in
the earlier example, and the second half having the same
preference order $(B,C,A)$ as agent $2$ in the earlier
example. If all agents bid truthfully, then the first $\ell$
agents each receive the sorted allocation $(1, 0, 1/2)$;
however if they lie and bid $(B,A,C)$, while the remainder
bid truthfully, then each lying agent receives the improved
sorted allocation $(1, 1/2, 0)$.
\end{example}

\section{Characterizing All Pareto Efficient Allocations}

Bogomolnaia and Moulin~\cite{BM01} extended their mechanism
by allowing players to receive
goods at time-varying rates. Specifically, for each agent
$i$ there is a speed function $\eta_i$ mapping
the time interval $[0,1]$ into the nonnegative reals, such
that for all $i$
\[ \int_0^1 \eta_i(t) \; dt = r_i. \]
Subject to these speeds, goods flow to agents in order of
the preference lists they bid, just as before. They showed
that this extension characterizes all ordinally efficient
allocations. 

In this section, we obtain an analogous
characterization of all Pareto efficient allocations by a
similar extension of our mechanism. Specifically, we prove
that for \textit{any} Pareto efficient allocation of bundles,
there exist speeds such that the extended SG mechanism produces
that allocation. We prove this after first noting that the
extended SG mechanism always results in Pareto efficient
allocations.

In this section when $\eta_i$ ($1 \leq i \leq n$) are fixed, we let
$a^\pi$ (with the $\eta$'s implicit) be the goods allocation produced by
the extended SG mechanism with these speeds and truthful bids. We let $l^\pi=L^\pi(a^\pi)$ be the corresponding allocation of bundles.

\subsection{Pareto Efficiency}
\begin{theorem} Let $\eta_i$, $1 \leq i \leq n$, be any
  speed functions. Then the allocation $l^\pi$ is
  Pareto efficient.
\end{theorem}
\begin{proof} The argument is the same as for Theorem~\ref{thm.Pareto.newproof} with the proviso that the definition 
$t_{iK}= \frac{1}{r_i} \sum_{k=1}^K l^\pi_{i \pi_i(k)}$
is replaced by 
$t_{iK}= \inf\{y: 
\int_0^y \eta_i(t) \; dt \geq \sum_{k=1}^K l^\pi_{i \pi_i(k)} \}$.
\end{proof}

\subsection{Characterizing All Pareto Efficient Allocations}
If the last result mirrored the First Welfare Theorem, the
next mirrors the Second Welfare Theorem:
\begin{theorem} Let $\pi$ be the collection of agent
  preference lists over bundles, and let $l$ be a Pareto efficient
allocation. There exist speed functions $\eta_i$,
  $1 \leq i \leq n$, such that $l=l^\pi$.
\end{theorem} \begin{proof} As before the bundles are $(\lambda^k)_{k=1}^M$, where for each $k$, $\sum_{j=1}^m \lambda^k_j = 1$, and $\lambda^k_j \geq 0$ for all $j$.

Construction of the speeds $\eta_i$ is simple. 
Let a ``partial bundle
allocation'' be a list $\hat{l}_{ik}$, each $\hat{l}_{ik} \geq 0$, such that for every $i$, $\sum_{k,j}
\hat{l}_{ik} \lambda^k_j \leq r_i$, and for every $j$, $\sum_{i,k} \hat{l}_{ik} \lambda^k_j  \leq q_j$.

Initialize $t=0$ and initialize each agent $i$ with the
empty partial allocation $\hat{l}_{ik}=0$ for all $i,k$.

Initialize $c_j$ to be the quantity of good $j$ that is allocated in $l$. (Necessarily $c_j \leq q_j$ and $\sum c_j=\sum r_i$. 
If $\sum q_j > \sum r_i$ then for some $j$, $c_j < q_j$.)

Then repeat the following until $t=1$.

Find an agent $i$ for whom there is an $\ell$ such that
$\hat{l}_{i\pi_i(\ell)}<l_{i\pi_i(\ell)}$, and such that for
all $\ell'<\ell$, the bundle $\pi_i(\ell')$ has been exhausted (that is to say, 
 there is a good $j$ such that $\lambda^{\pi_i(\ell')}_j>0$ and $c_j=0$.)
To see that there is such an $i$, suppose the contrary, and consider all the agents for whom 
 $\sum_{k,j}
\hat{l}_{ik} \lambda^k_j < r_i$. For each of them there is a favorite bundle which has not yet been exhausted. Evidently none of these agents is to be allocated in $l$ any additional quantity of this favorite bundle. However since these favorite bundles have not yet been exhausted, we can allocate to every player a slight additional positive amount of his favorite unexhausted bundle, without exhausting any additional goods. Any extension of this new partial bundle allocation to a full bundle allocation, strictly Pareto dominates $l$, contrary to assumption.

Now set $\delta=
(l_{i\pi_i(\ell)}-\hat{l}_{i\pi_i(\ell)})/\sum r_i$.  For $t
< t' < t+\delta$, make the settings $\eta_i(t')=\sum r_i$
and, for $i' \neq i$, $\eta_{i'}(t')=0$. Then increment
$\hat{l}_{i\pi_i(\ell)}$ by $\delta \sum r_i$, and decrement each $c_j$ by the corresponding amount, namely, decrement $c_j$ by
$\lambda^{\pi_i(\ell)}_j \delta \sum r_i$. 
 Finally, increment $t$
by $\delta$.

\suppress{ 
This process can only fail to complete if there comes a time
$t$ at which every agent $i$ falls into one of the following
sets $S_1$ and $S_2$, and $S_1$ is nonempty:
\begin{enumerate}
\item $S_1= \{i$ such that $\sum_j \alpha_{ij} < r_i$, and
the minimal $\ell$ for which $c_{\pi_i(\ell)}>0$ also
satisfies $\alpha_{i\pi_i(\ell)}=a_{i\pi_i(\ell)} \}$.
\item $S_2= \{i$ such that $\sum_j \alpha_{ij} = r_i\}$.
\end{enumerate}
If there is such a $t$, then
for each $i \in S_1$, define $\ell_i$ to be the $\ell$
identified in the definition of $S_1$.
Then for each $i \in S_1$, a small positive amount of the good
$\pi_i(\ell_i)$
can be added to $i$'s current partial allocation; the new
partial allocation, no matter how it is completed to an
allocation, improves strictly on $a$ for all $i \in S_1$,
and is unchanged for $i \in S_2$. Therefore $a$ is not
Pareto efficient.
}
This process terminates in finitely many iterations because in each iteration some agent completes its allocation of some bundle.
\end{proof}

Examination of the above proof reveals:
\begin{corol}
There is a polynomial time algorithm for checking whether a given
allocation is Pareto efficient.
\end{corol}

\subsection{No Incentive Compatibility for the Variable Speeds
  Variant}

We note that the synchrony imposed among agents by the SG
mechanism is key to its incentive compatibility and
envy-freeness properties (indeed, the properties hold even
if the basic mechanism is extended with the {\em same} speed
function for all agents). If different agents have
different speed functions under the extended SG mechanism,
Theorems \ref{thm.incentive-compatible} and
\ref{thm.gruppen}, showing incentive compatibility, fail to
hold. The argument breaks down as soon as it uses
termination times, in Lemma~\ref{timing-lemma}.  Below is a
counter-example for strategy-proofness; a similar idea gives
counter-examples for group strategy-proofness and
envy-freeness.

\begin{example} Assume
$m=n=4$ and that all $r_i=q_j=1$. Let the speed function for
agent $1$ be $1$ over the interval $[0, 1]$. The speeds of
agents 2, 3, and 4 equal $1$ over the interval $[0, 1/2]$,
$0$ over the interval $(1/2,5/6]$, and $3$ over the interval
$(5/6,1]$. The preference orders of agents $1$ and $2$ are
$(1, 2, 3, 4)$, and the preference orders of agents $3$ and
$4$ are $(2, 4, 3, 1)$. If all agents bid truthfully, agent
$1$ receives the sorted allocation $(1/2, 0, 1/2, 0)$. On
the other hand, if agent $1$ bids $(2, 1, 3, 4)$ while the
rest bid truthfully, then agent $1$ receives the better
sorted allocation $(1/2, 1/3, 1/6, 0)$. \end{example}

\section{Envy-Freeness w.r.t.\ stochastic dominance preference}
(This section is the only part of the paper where we use sd preference.)

Given a bundle allocation $l$, let $\bar{l}$ denote the
\textit{relative allocation,} where
$\bar{l}_{ij} = l_{ij}/r_i$.
\begin{theorem} Under truthful bidding, every agent $i$ weakly
sd-prefers his relative allocation $\bar{l}^\pi_{i*}$ to the
relative allocation $\bar{l}^\pi_{i'*}$ of any other agent $i'$. \end{theorem} 
\begin{proof} Fix any $1 \leq k \leq M$. We are to show that
\[ \frac{1}{r_i} \sum_{\ell=1}^k l^\pi_{i\pi_{i}(\ell)} \geq \frac{1}{r_{i'}}
\sum_{\ell=1}^k l^\pi_{i'\pi_{i}(\ell)} .\]
Let $t$ be the time at which the last of the bundles $\pi_i(1),\ldots,\pi_i(k)$ is exhausted. So $t r_i=\sum_{\ell=1}^k l^\pi_{i\pi_{i}(\ell)}$.
No other agent
can receive any of these bundles after time $t$, so $t r_{i'} \geq
\sum_{\ell=1}^k l^\pi_{i'\pi_{i}(\ell)}$.
\end{proof}

\section{Other Greedy Mechanisms} 
\label{counter}

As stated in the Introduction, obtaining an efficient and envy-free non-pricing 
mechanism for allocating divisible goods is easy, but additionally satisfying incentive compatibility 
is harder. In this section we present two greedy mechanisms which satisfy the first two properties
but not the third. To simplify description of the mechanisms, assume that $m=n$ and that all $r_i=q_j=1$; it is straightforward to generalize the mechanisms beyond this restriction, and our counterexamples are possible even with it.

{\bf Mechanism 1:} The mechanism proceeds iteratively. In round  $i$, it considers the $i$th-favorite goods
of all agents who still have not been allocated a full unit of goods. Among such agents, if the $i$th-favorite good of a set $S$
of agents is good $j$, the remaining quantity of good $j$ is allocated equally among the agents in $S$, subject to no agent getting more than a total of one unit of goods. (Some of good $j$ may remain after the round.)

{\bf Mechanism 2:} The mechanism has a notion of time, similar to SG. Goods allocation starts at time $0$ and is completed at time $1$. During this interval each agent receives goods at rate $1$. The interval is punctuated by finitely many critical instants at which some of the agents switch which good they are receiving. The first critical instant is $0$ and the others are the times at which some nonempty set of agents $T$ finishes receiving their promised allocation of a good. At such an instant, the mechanism identifies, for each of the agents in $T$, the next-favorite good on their list that has not yet been fully promised to other agents. The mechanism promises each agent in $T$ some of that good, in the following fashion: let $T_j$ be the subset of $T$ requesting good $j$ and let $u$ be the amount of good $j$ that has not been previously promised. Then each agent in $T_j$ is promised an equal share of $u$ subject to no agent exceeding a total of one unit of goods. (The next critical instant affecting these agents is of course easily computed.) The mechanism then proceeds to the next critical instant.

The proofs given above, for showing that the SG mechanism is efficient and envy-free, extend easily to showing that Mechanisms
1 and 2 are also efficient and envy-free. Here, however, are counterexamples to incentive compatibility:

\begin{example} \textbf{Mechanism 1:} Let
$m=n=4$; name the goods $A, \ldots, D$.
Agent 1's preference list is $A,B,C,D$; 
agents $2$ and $3$ have preferences $A,C,B,D$; and agent $4$'s favorite good is $B$. If the agents bid truthfully then in round $1$, agent $4$ is allocated all of good $B$, while the first three agents are each allocated a third of good $A$. In the second round agent $1$ is left out while agents $2$ and $3$ are allocated half of good $C$. In round $3$ no allocations are made, and in round $4$ good $D$ is allocated among the first three agents. The allocation to agent $1$ is therefore $(A: 1/3, D: 2/3)$. If instead agent $1$ submits the preference list $A,C,B,D$ then she is treated the same as agents $2$ and $3$, and her allocation is $(A:1/3, C: 1/3, D: 1/3)$, which she prefers. \end{example}

The counterexample for the second mechanism is more involved.

\suppress{
\begin{example} \textbf{Mechanism 2:} \ljs{Bug in this one:} Let
$m=n=8$; name the goods $A, \ldots, H$. We specify only the essential components
of the preference orders.
The preference order of agent $1$ is $(A, D, C, H, \ldots)$, of agent 5 is $(B, D, \ldots)$ of agents 2, 3, 4 is
$(A, F, G, H, \ldots)$, and of agents 6, 7, 8 is $(B, C, E, F, G, H, \ldots)$.
If all agents report their preferences truthfully, then agent 1 gets the allocation $1/4$ each of $A$ and $D$, \ljs{Actually he gets $1/2$ of $D$} and $1/2$ of $H$.
On the other hand, if agent 1 lies and reports the order $(A, C, D, H, \ldots)$, then she gets the improved allocation of
$1/4$ each of $A, C, D$ and $H$,
\end{example}

\ljs{Revised. I think this is closer to what you described on the phone:}
}

\begin{example} \textbf{Mechanism 2:} Let
$m=n=8$; name the goods $A, \ldots, H$. We specify only the essential components
of the preference orders.
The preference order of agent $1$ is alphabetical, $(A, \ldots, H)$. Agents $2,3,4$ have the preference order $(A,G,H,F,\ldots)$. Agents $5,6,7$ have the preference order $(B,C,E,F,\ldots)$. Agent $8$ has the preference order $(B,D,\ldots)$. If all agents 
report their preferences truthfully, agent 1 gets the allocation $(A:1/4, C: 1/4, D: 1/4, F: 1/4)$; if agent $1$ lies and reports the order $(A,C,E,D,\ldots)$ she gets the allocation $(A:1/4, C: 1/4, D: 1/4, E: 1/4)$, which she prefers.
\end{example}

\section{Discussion}
Our main open problem is the one mentioned in the
Introduction, i.e., achieving approximate versions of the
properties of the SG mechanism but when agents' preferences
are representable by utility functions.


\suppress{
A natural question, especially in view of the mechanisms provided in Section \ref{counter}, is whether
SG is the {\em unique} mechanism that is efficient, incentive compatible and envy-free. A related result, of
Bogomolnaia and Heo \cite{BH12}, is that efficiency
(under the sd relation), envy-freeness, and a
property they call bounded invariance characterize the PS mechanism. 
A mechanism is
said to have the \textit{bounded invariance property} if for
any agent $i$ and any good $j$, changing $i$'s preference
order for goods she likes less than $j$ does not change the
amount (equivalently, probability) of good $j$ each agent gets.
}

\suppress{
This leads to the question of appropriately
characterizing the SG mechanism. Towards this end we ask if
efficiency, under the more stringent lexicographic relation,
and envy-freeness suffice. Clearly, a first step would be to
characterize the PS mechanism in this manner. }

Another natural open question concerns the existence of mechanisms
to produce lexicographically most equitable allocations,
having favorable algorithmic and game-theoretic properties
(esp., incentive compatibility). The SG mechanism is not very equitable: see Appendix~\ref{sec.equitable}. 


\suppress{
As mentioned in the Introduction,
Katta and Sethuraman~\cite{KS06} generalize the setting of
Bogomolnaia and Moulin to the ``full domain'', i.e., agents
may be indifferent between pairs of goods, and they give a
randomized mechanism that is a generalization of PS and achieves
the same game-theoretic properties as PS. We ask the analogous
question for the generalization of our setting, i.e., each agent
partitions the goods by equality and defines a total order
on the equivalence classes of her partition (the agent is
equally happy with any good received from an equivalence
class). Preferences are again defined lexicographically over
classes, i.e., by considering the total amount of goods
received in each class. Is there a generalization of our
mechanism to this setting?
}

Finally, one expects that allocation quality will increase with the
diversity of agent preferences. A natural waystep to consider is the
very diverse setting in which agent preferences are independent
uniformly random permutations of $1, \ldots, m$. This suggests many 
interesting questions, e.g., what is the distribution of
$\beta_k$ as in Eqn.~\ref{beta-eqn},
for the allocation $a$ given by the SG mechanism; and,
what is the distribution of the maximized value of $t$ as in LP
(\ref{lp.equitable}), both for various values of $k$. Regarding
the latter question, for $n=m$ and all $r_i=q_j=1$, there is
a correspondence in the case $k=1$ with the collision
statistics of random pointers, and so it is known that $t
\to (\log \log n)/(\log n)$; for larger $k$ there is a rough
correspondence with the ``power of two choices''
literature~\cite{AzarBKU99,MRSchapter}, suggesting likely
asymptotics of $(\log k)/(\log \log n)$ for fixed $k$,
although the correspondence between the problems is not
close enough for us to state this with certainty.

\section{Acknowledgments}
We are indebted to Herv\'{e} Moulin for generously sharing his deep
understanding of this research domain. Thanks also to
Jeremy Hurwitz for stimulating discussions and to Amin
Saberi for pointing us to useful references.

Schulman was supported in part by NSF Grants 1038578 and 1319745. Part of this work was done while visiting the Simons Institute for the
Theory of Computing at UC Berkeley. 

Vazirani was supported in part by NSF Grants CCF-0914732 and CCF-1216019, and 
a Google Research Grant. He would like to thank the Social and Informational Sciences Laboratory at Caltech for their hospitality; part of this work was done while he was Distinguished SISL Visitor during 2011-12. He would also like to thank the Guggenheim Foundation for a Fellowship during 2011-12.

\bibliographystyle{plain}
\bibliography{refs}

\appendix
\section{Equitability of Allocations}
\label{sec.equitable}

It is interesting to consider whether an allocation mechanism is
equitable---minimizing, in some measurable sense, the
disparity in the welfare of the players. In spite of being
deterministic and treating all agents symmetrically, the SG
mechanism is not particularly equitable, except as regards how much each agent receives of his most-preferred good.
We provide an example showing that even the allocations of each
agent's two most preferred goods may be quite
inequitable. 

On the other hand we describe a
time-efficient algorithm that, for any given $1 \leq k \leq
m$, equitably allocates the top $k$ goods for each agent. We
further define the notion of a \textit{lexicographically
most equitable} allocation and give a time-efficient
algorithm to find one. In this section we consider only the non-Leontief case.

Recall that for an allocation $a$ let $\bar{a}$ denotes the
\textit{relative allocation}
$\bar{a}_{ij} = a_{ij}/r_i$. For any $k, \ 1 \leq k \leq m$,
say that an allocation is \textit{equitable w.r.t.\ agents'
  top $k$ choices} if it belongs to
\[ \arg \! \max_{a}  \min_i  {(\bar{a}_{i \pi_i(1)} + \ldots
  +  \bar{a}_{i \pi_i(k)})}   ,\]
where the max is over all allocations $a$.

It is easy to see that the allocation produced by the SG
mechanism is equitable for $k = 1$. However, as the
following example illustrates, it is not equitable for $k=2$, or
larger values of $k$.

\begin{example} \label{not-equitable} Let $n = 2$, $m = 3$, $r_1=r_2=1$,
$q_1=1/2$, $q_2=5/6$, and $q_3=2/3$.  Let the preference
list of the first agent be $(1, 2, 3)$ and that of the
second agent $(2, 3, 1)$. Then the SG mechanism gives sorted
allocations of $(1/2, 1/6, 1/3)$ and $(2/3, 1/3, 0)$
respectively to the agents, so each receives $2/3$ of his
total allocation from his top two choices. On the other
hand, the sorted allocations $(1/2, 1/2, 0)$, $(1/3, 2/3,
0)$ are also feasible, and in this case each agent receives
his entire allocation from his top two
choices. \end{example}

Next, we show that there is a polynomial-time algorithm
which, given $k$, $(r_i)$, $(q_j)$ and the list of agent
preferences, obtains an allocation that is equitable
w.r.t.\ agents' top $k$ choices. In fact this allocation
$x=(x_{ij})$ is the solution to the linear program given
below, together with $t$, the minimum over agents of the
relative allocation from the agent's top $k$ goods.

\begin{align}
\label{lp.equitable}
\textrm{Maximize \ } &  t  \\ \smallskip
\textrm{Such that \;} & \forall i: \;   t \leq
\frac{1}{r_i} \left(\sum_{\ell =1}^k {x_{i \pi_i(\ell)} }
\right)  \nonumber  \\
 & \forall i: \; \sum_{j=1}^m {x_{ij} } =  r_i  \nonumber  \\
          & \forall j:  \; \sum_{i=1}^n x_{ij} \leq  q_j  \nonumber \\
          & \forall i \ \forall j: \; x_{ij} \geq 0 \nonumber
\end{align}

Finally, let us define the notion of the
\textit{lexicographically most equitable allocation}, which
intuitively is an allocation that simultaneously optimizes
for each $k$, to the extent possible. For any allocation
$a$, and each $k, \ 1 \leq k \leq m$, define
\be \beta_k = \min_i  {(\bar{a}_{i \pi_i(1)} + \ldots +
  \bar{a}_{i \pi_i(k)})}   . \label{beta-eqn} \ee
Now, define a \textit{lexicographically most equitable
  allocation} to be one that lexicographically maximizes
$(\beta_1, \ldots, \beta_m)$.

We now give a polynomial-time algorithm to find a
lexicographically most equitable allocation---it involves
solving $m$ LPs derived from LP (\ref{lp.equitable}). The
first LP simply computes $\beta_1$ by solving LP
(\ref{lp.equitable}) for $k = 1$. Next, for each $k, \ 2
\leq k \leq m$, add the following constraints to LP
(\ref{lp.equitable}) and solve it to determine $\beta_{k}$:
\[  \forall i, \forall 1 \leq h \leq k-1: \    \frac{1}{r_i}
\left(\sum_{\ell =1}^h {x_{i \pi_i(\ell)} } \right)  \geq
\beta_{h} .\]
Clearly, the last LP will yield a most equitable allocation.

\begin{example} For the agents in
  Example~\ref{not-equitable}, the
lexicographically most equitable allocation is (given as a
sorted allocation): $(1/2, 1/3, 1/6)$ for agent $1$ and
$(1/2, 1/2, 0)$ for agent $2$. This is different from both
the SG allocation and the allocation that is equitable w.r.t.\ agents' top $2$ choices. \end{example}

Although equitability would seem to be a desirable property, it
must be noted that an equitable allocation need not be even
Pareto efficient:

\begin{example} Let $n=3$, $m=4$, $r_1=r_2=r_3=2$,
  $q_1=q_2=1$, $q_3=q_4=2$. Let the preference lists be
  $\pi_1=(1,2,3,4)$, $\pi_2=(3,4,1,2)$,
  $\pi_3=(4,3,1,2)$. For any $0 \leq x \leq 1$, the
  following allocation is lexicographically most equitable,
  and even stronger, it simultaneously optimizes all
  $\beta_k$ in Eqn.~\ref{beta-eqn}:
\[ a_1=(1,1,0,0), a_2=(0,0,2-x,x), a_3=(0,0,x,2-x). \]
Yet this allocation is Pareto efficient only in the
single case $x=0$.
\end{example}
\end{document}